\documentclass[11pt,a4paper]{amsart}

\usepackage[a4paper,margin=1.5cm]{geometry}
\usepackage{amsmath,amssymb,amsthm}
\usepackage{amscd}
\usepackage{graphicx}
\usepackage{booktabs}
\usepackage{array}
\usepackage{longtable}
\usepackage{float}
\usepackage{placeins}
\usepackage{cite}
\usepackage{pdflscape}
\usepackage{adjustbox}
\usepackage{xcolor}
\usepackage{chngcntr}
\usepackage{tikz}
\usepackage{placeins}
\usetikzlibrary{arrows.meta,positioning}

\counterwithout{figure}{section}
\counterwithout{table}{section}

\newtheorem{theorem}{Theorem}
\newtheorem{proposition}[theorem]{Proposition}

\title[Evolutionary Entropy and the Greenland Shark]
{Evolutionary Entropy Shapes Lifespan of the Greenland Shark}

\author[J. Buescu]{Jorge Buescu}
\address{Departamento de Matemática, Faculdade de Ciências,
Universidade de Lisboa, Portugal;
Centro de Estudos Matemáticos, Universidade de Lisboa, Portugal}
\email{jsbuescu@ciencias.ulisboa.pt}

\author[S. Elaydi]{Saber N. Elaydi }
\address{Department of Mathematics, Trinity University,
San Antonio, Texas, USA; Fellow of the Center for Mathematical Analysis,
Geometry and Dynamical Systems (CAMGSD),
University of Lisbon, Portugal}
\email{selaydi@trinity.edu}

\author[H. M. Oliveira]{Henrique M. Oliveira$^*$}
\address{Departamento de Matemática, Instituto Superior Técnico,
Universidade de Lisboa, Portugal;
Centro de Análise Matemática, Geometria e Sistemas Dinâmicos,
Universidade de Lisboa, Portugal}
\email{henrique.m.oliveira@tecnico.ulisboa.pt}
\thanks{$^*$Corresponding author.}



\date{\today}

\begin{document}

\begin{abstract}
The Greenland shark (\emph{Somniosus microcephalus}) is among the
longest-lived vertebrates known, with female maturity estimated at
about 156 years and maximum ages approaching four centuries. We use
the entropic theory of life histories and an open-group Leslie model
to determine which reproductive windows are compatible with this
extreme demographic profile. The model yields generation times,
reproductive quantiles and entropy-maximizing reproductive endpoints.
A critical value of reproductive persistence, \(q_c\approx0.9763\),
separates two regimes: below it, evolutionary entropy increases
towards an asymptote; above it, entropy is maximized at a finite
reproductive endpoint. Longevity-calibrated scenarios lie on both
sides of this threshold and imply generation times near or above two
centuries, with reproductive contributions extending into the third,
fourth and fifth centuries of life. As an independent benchmark, an empirically parametrized Leslie model
for the bowhead whale yields comparable entropy-based demographic
quantities while allowing the reproductive and survivorship parameters
to be calculated directly. The method provides quantitative
constraints on reproductive organisation in exceptionally long-lived
species for which detailed demographic data are unavailable.
\end{abstract}

\maketitle

\section{Introduction}

The Greenland shark, \emph{Somniosus microcephalus}, is among the
longest-lived vertebrates currently known. 
Radiocarbon dating of
eye-lens nuclei suggests ages exceeding three centuries, with the
oldest individual estimated at $392 \pm 120$ years, while
female sexual maturity has been estimated at
$156 \pm 22$ years \cite{Nielsen2016GreenlandShark}.
%
These observations raise a natural demographic question: what reproductive window is compatible with such an extreme delay in maturity?

We address this question within the entropic theory of life histories
introduced by Demetrius
\cite{Demetrius1974,Demetrius1977}.

\begingroup
\color{black}
Classical life-history theory is an evolutionary theory of age-specific
survivorship and fecundity. In age-structured populations, the asymptotic
population growth rate provides a demographic measure of Darwinian fitness,
while the broader selection-theoretic background includes Fisher's
Fundamental Theorem of Natural Selection \cite{Fisher1930}. This classical
tradition was developed into a systematic theory of life histories by, among
others, Charlesworth, Stearns, Roff, and Charnov
\cite{Charlesworth,Stearns1992,Roff1992,Charnov1993}. It describes how
natural selection organizes demographic schedules and the associated
trade-offs involving age at maturity, age-specific survival and fecundity,
reproductive output, and lifespan.

The entropic theory of life histories is a
bioenergetic and statistical theory of the evolution of age-specific
survivorship and fecundity
\cite{Demetrius1974,Demetrius1977,Demetrius97,Demetrius2000}. In its
bioenergetic interpretation, the population growth rate measures the net rate
at which the chemical energy supplied by environmental resources is converted
into biological work, and admits the decomposition
\begin{equation}
\label{eq:entropy-decomposition}
r=H+\Phi.
\end{equation}
Here, in discrete time, $r=\log\lambda$, where $\lambda$ is the Perron \textcolor{black}{eigenvalue} of the matrix modelling the dynamics. The function $H$, called
evolutionary entropy, describes the rate associated with the conversion of
metabolic energy into biological work over the life cycle, whereas $\Phi$,
called reproductive potential, represents the complementary contribution of
the resource environment to demographic growth
\cite{Demetrius97,Demetrius2000,Demetrius09,
DemetriusLegendreHarremoes2009,Demetrius2023SelfOrganization}.

The quantities $H$ and $\Phi$ also specify statistical properties of the life
history. Evolutionary entropy is
\[
H=\frac{S}{T},
\]
where $S$ is the entropy of the normalized distribution of reproductive
contributions among the age or stage classes and $T$ is generation time.
Thus, $S=0$ when reproductive contribution is concentrated in a single class,
as in the idealized semelparous case, whereas $S>0$ when reproduction is
distributed over several classes, as in iteroparous life histories
\cite{Demetrius1974,Demetrius1977,Demetrius97,Demetrius2000}.

The reproductive potential $\Phi$ orders life histories along the fast--slow
continuum. Negative values, $\Phi<0$, are associated with the slow end of the
continuum---late sexual maturity, long lifespan, and low reproductive output
or small litter size---whereas positive values, $\Phi>0$, are associated with
the fast end---early sexual maturity, short lifespan, and high reproductive
output or large litter size
\cite{Demetrius97,Demetrius2000,Demetrius09,
Demetrius2023SelfOrganization}.

The entropic selection principle relates the direction of selection on $H$ to
the external environment, represented by the resource endowment: stable
resource conditions favour increasing evolutionary entropy, whereas unstable
resource conditions favour decreasing evolutionary entropy. Evolutionary entropy also has a precise dynamical interpretation:
after transforming the Leslie model into an equivalent Markov chain,
$H$ determines the asymptotic rate of convergence towards the stable
age distribution
\cite{Tuljapurkar1982}. Larger values of $H$ are also associated
with greater demographic stability against demographic fluctuations
\cite{Demetrius04,DemetriusGundlach2014}.

The pair $(H,\Phi)$
therefore links the statistical organization of a life history to the resource
environment and provides quantitative descriptors of its temporal and
functional organization. The Greenland shark offers an extreme application of
this theory: its exceptionally delayed maturity, extended lifespan, and low
reproductive output place it near the slow end of the life-history continuum
\cite{Nielsen2016GreenlandShark,Gravel2024}.
\endgroup

The analysis relies on two results established in a recent paper
\cite{BuescuElaydiOliveira2026}. The \textcolor{black}{Homogeneity} Theorem 2.4 states
that evolutionary entropy, generation time and reproductive quantiles
depend only on the normalized reproductive distribution. In the same
paper it is shown that there exist two qualitatively different regimes
for evolutionary entropy in the reproductive window: either the entropy
always increases to a horizontal asymptote, or it has a unique interior
maximum.

The \textcolor{black}{Critical-Threshold} Theorem 4.14 states that these regimes, in
open-group Leslie models, are separated by the unique solution of the
criticality equation
\begin{equation}
\label{eq:critical}
q^A+q-1=0,
\end{equation}
where \(A\) is the age at first reproduction and \(q=\rho/\lambda\),
with \(\rho\) denoting the asymptotic survival parameter of the open
class. In Section~\ref{sec:sensitivity} we show that, within the
geometric open-group family, equation~\eqref{eq:critical} is exactly
the stationarity condition \(dH/dq=0\). By Theorem 4.10 of
\cite{BuescuElaydiOliveira2026}, this stationary point is the unique
maximizer of \(H(q)\). The critical threshold therefore has a direct
variational interpretation.


In this paper we consider open-group populations with reproductive persistence
close to the critical threshold, including both subcritical and supercritical
scenarios. For each scenario we compute the generation time, reproductive quantiles
and, when it exists, the entropy-maximizing reproductive endpoint. 
We also derive the asymptotic divergence of the optimal reproductive
window as the persistence parameter approaches the critical threshold
from above.

For comparative context, in Section~5 we apply the same construction
to five other exceptionally long-lived animal species, showing that
the Greenland shark occupies a distinguished demographic position even
within this extreme-longevity group. For the bowhead whale,
\textit{Balaena mysticetus}, an independently parametrized Leslie
model additionally provides an empirical benchmark against which the
reduced open-group construction can be compared. Our goal is not to
reconstruct a demographic matrix for \textit{S.~microcephalus}, but
to determine which reproductive windows are compatible with a maturity
age of 156 years.

Throughout the analysis, ages and projection intervals are expressed
in years.

\section{Open-Group Representation}\label{Open-Group Representation}

The demographic scenarios considered below may be represented by an
open-group Leslie matrix \cite{ElaydiCushing2024}
\[
L_q=
\begin{pmatrix}
0 & \cdots & 0 & f & f\\
s_1 & \cdots & 0 & 0 & 0\\
\vdots & \ddots & \vdots & \vdots & \vdots\\
0 & \cdots & s_{A-1} & 0 & 0\\
0 & \cdots & 0 & s_A & \rho
\end{pmatrix},
\]
\textcolor{black}{where $A$ denotes the age at first reproduction. Here, ``open group''
refers to a terminal aggregate age class. In a strictly age-truncated Leslie matrix, individuals
surviving beyond the last represented age are no longer followed.
In $L_q$, by contrast, all ages beyond $A$ are collected in the
terminal class $A^{+}$. The coefficient $s_A$ transfers individuals
into this class, while the lower-right entry $\rho$ represents their
survival within it from one projection interval to the next. The
self-loop of weight $\rho$ in Figure~\ref{fig:open-group-graph}
therefore permits an indefinitely long post-maturity lifespan without
imposing an arbitrary maximum age.}
\begin{figure}[H]
\centering
\resizebox{0.50\textwidth}{!}{%
\begin{tikzpicture}[
  state/.style={
    circle,
    draw,
    minimum size=8mm,
    inner sep=1pt
  },
  transition/.style={
    -{Latex[length=2mm,width=1.4mm]},
    thick
  },
  node distance=1.45cm
]

\node[state] (x1) {$1$};
\node[state, right=of x1] (x2) {$2$};
\node[draw=none, right=1.10cm of x2] (dots) {$\cdots$};
\node[state, right=1.10cm of dots] (xAm1) {$A-1$};
\node[state, right=of xAm1] (xA) {$A$};
\node[state, right=of xA] (xOpen) {$A^{+}$};

\draw[transition]
  (x1) -- node[above] {$s_1$} (x2);

\draw[transition]
  (x2) -- node[above] {$s_2$} (dots);

\draw[transition]
  (dots) -- node[above] {$s_{A-2}$} (xAm1);

\draw[transition]
  (xAm1) -- node[above] {$s_{A-1}$} (xA);

\draw[transition]
  (xA) -- node[above] {$s_A$} (xOpen);

\draw[transition]
  (xA.north)
  to[out=120,in=60,looseness=1.15]
  node[above] {$f$}
  (x1.north);

\draw[transition]
  (xOpen.south)
  to[out=240,in=300,looseness=1.10]
  node[below] {$f$}
  (x1.south);

\draw[transition]
  (xOpen)
  edge[loop right]
  node[right] {$\rho$}
  (xOpen);

\end{tikzpicture}
}

\caption{Life-cycle graph associated with the open-group Leslie
matrix $L_q$. The class $A^{+}$ is the terminal open class containing
individuals older than age $A$.}
\label{fig:open-group-graph}
\end{figure}
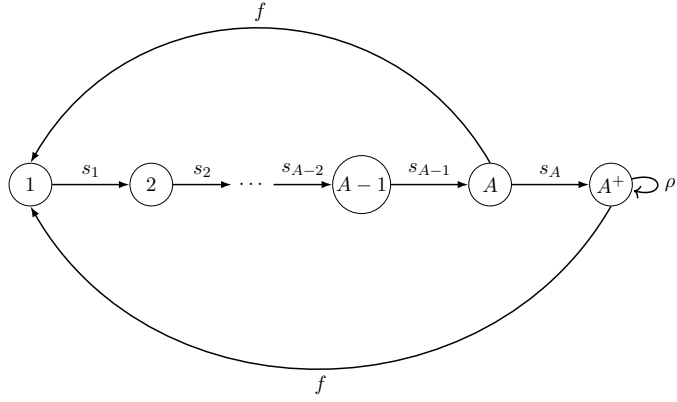
As established in \cite{BuescuElaydiOliveira2026}, reproductive contributions
follow a geometric tail of ratio $
q=\rho/\lambda,$ where
$\lambda$ is the dominant eigenvalue of the projection matrix. Under constant post-maturity fertility, the normalized reproductive
contributions form a geometric distribution with ratio \(q=\rho/\lambda\)
and normalization factor \(1-q\).

\textcolor{black}{For this geometric family, we use throughout the formulae derived in
\cite{BuescuElaydiOliveira2026}.} In particular, the normalized
reproductive-age distribution is
\[
p_j=(1-q)q^{j-A},
\qquad j\geq A.
\]
\textcolor{black}{The quantity \(S\) and the generation time
\(T\) are, respectively,}
\[
\begin{aligned}
S(q)
&:=-\sum_{j=A}^{\infty}p_j\log p_j
 =-\log(1-q)-\frac{q}{1-q}\log q,
\\
T(q)
&:=\sum_{j=A}^{\infty}j\,p_j
 =A+\frac{q}{1-q}.
\end{aligned}
\]
\textcolor{black}{The entropy \(S\) has the same mathematical form as Shannon entropy and
Gibbs--Boltzmann entropy, but the underlying probabilities and the
interpretations are different. Here, \(p_j\) is the normalized contribution
of age or stage class \(j\) to reproduction, so \(S\) measures the statistical
dispersion of reproductive contributions across the life cycle. Evolutionary
entropy is the corresponding rate \(H=S/T\), obtained by normalizing $S$ by generation time. It is therefore a demographic quantity,
not the thermodynamic entropy of the organism, although the common
logarithmic form permits the bioenergetic interpretation developed in
\cite{Demetrius97,Demetrius2000,Demetrius2023SelfOrganization}.}
Hence the evolutionary {entropy} of the geometric open-group family is
\begin{equation}
\label{eq:Hq}
H(q)=\frac{S(q)}{T(q)}
=
\frac{-\log(1-q)-\dfrac{q}{1-q}\log q}
     {A+\dfrac{q}{1-q}}.
\end{equation}
\textcolor{black}{ Combining equation~\eqref{eq:Hq} with the entropy/potential decomposition
\eqref{eq:entropy-decomposition}, we obtain explicitly the reproductive
potential \(\Phi=\log\lambda-H(q)\).
Thus \(H\) is determined by \(A\) and \(q\), whereas the reproductive
potential \(\Phi\) additionally requires the population growth factor
\(\lambda\).} Since
\[
q=\frac{\rho}{\lambda},
\]
knowledge of \(q\) alone does not identify \(\lambda\) unless the
annual persistence parameter \(\rho\) is independently known.

For comparative context, Gravel et al.~\cite{Gravel2024} report a
maximum intrinsic rate of population increase
\(r_{\max}\simeq 0.04\,{\rm yr}^{-1}\) for the Greenland shark, placing
the species at the extreme slow end of the fast--slow life-history
continuum. Related demographic formulations for chondrichthyans have
been developed by Pardo et al.~\cite{Pardo2016}, with explicit account
of survival to maturity. These estimates provide valuable independent
life-history context, although \(r_{\max}\) need not be identified
directly with the Perron growth rate \(\log\lambda\) of the reduced
open-group model considered here.

Indeed, since \(0<\rho\leq 1\),
\[
q\lambda=\rho\leq 1,
\]
and hence
\begin{equation}
\label{eq:growth-bound}
\log\lambda\leq -\log q.
\end{equation}
For the central longevity-calibrated value \(q=0.987413\), this gives
\[
\log\lambda\leq 0.01267\,{\rm yr}^{-1}.
\]
\textcolor{black}{Thus \(H(q)\), but neither \(\lambda\) nor \(\Phi\), can be calculated
from \(A\) and \(q\) alone. Moreover,
\begin{equation}\label{eq:super}
\Phi=\log\rho+T(q)^{-1}\log\left(\frac{1-q}{q^A}\right),
\end{equation}
so \(q>q_c\) implies \(\Phi<0\) for every \(0<\rho\leq1\), whereas for
\(q<q_c\) the sign of \(\Phi\) depends on \(\rho\).}

The open terminal class and the assumption of constant post-maturity
fertility play distinct roles. Empirically, open-group structure is
common: as reported in \cite{BuescuElaydiOliveira2026}, more than
70\% of the Leslie matrices examined from the COMADRE database
\cite{SalgueroGomez2016COMADRE} are open-group matrices. By contrast,
constant post-maturity fertility is a deliberate analytical restriction.
It isolates the effect of reproductive persistence and yields the
geometric kernel studied here. Age-dependent fertility would lead to a
weighted reproductive tail and lies beyond the reduced analytical model
considered in this note.


\textcolor{black}{Biologically, $q = \rho/\lambda$ measures the persistence of reproductive contributions relative to population growth. 
Values near 1 indicate high persistence of reproductive contributions
relative to population growth and, when population growth is close to
stationarity, correspond to high adult survival and slow turnover. This pattern is expected under low extrinsic mortality, cold stable
environments, and depressed metabolism. The Greenland shark,
inhabiting near-freezing Arctic waters and occupying a high trophic
position, appears to face little documented natural predation as an
adult \cite{MacNeilEtAl2012,GrantEtAl2018}, making it a natural
candidate for such near-unity persistence.}

\section{Entropy-Based Demographic Scenarios}

Open-group Leslie models provide a natural representation for species
characterized by extremely long reproductive tails
\cite{ElaydiCushing2024}.
Adopting for the Greenland shark the 
estimate $A=156$ years,  
we obtain 
 from the criticality equation \eqref{eq:critical} the value  
 $
q_c\approx 0.9763$ for the critical threshold. 

To investigate demographic organizations compatible with the observed life
history of \textit{Somniosus microcephalus}, we calibrate three values of \(q\) by requiring the \(95\%\) reproductive
quantile \(B_{95}\) to match the lower, central, and upper longevity
estimates of \(272\), \(392\), and \(512\) years, respectively. 
The 95\% quantile is used for calibration because it provides a stable
estimate of the reproductive endpoint while being only weakly affected
by the long geometric tail, which is determined by the model rather
than by the data.

For $0<p<1$, let $B_p$ denote the smallest age at which the
cumulative reproductive contribution reaches $p$. From the geometric
reproductive-age distribution,
\begin{equation}
B_p
=
A+\left\lceil
\frac{\log(1-p)}{\log q}
\right\rceil-1.
\label{eq:quantile-age}
\end{equation}
Consequently, for a prescribed target age $B>A$,
\begin{equation}
B_p=B
\quad\Longleftrightarrow\quad
(1-p)^{1/(B-A)}
<
q
\leq
(1-p)^{1/(B-A+1)}.
\label{eq:quantile-interval}
\end{equation}
Thus each discrete target age determines an interval of admissible
values of $q$, rather than a unique value. In the calculations below,
$q$ is taken as the arithmetic midpoint of the corresponding
interval. We obtain
\[
q=0.974613,\;0.987413,\;0.991632.
\]
The first scenario lies slightly below the critical threshold \(q_c\),
whereas the central and upper scenarios lie above it.
For each value we compute the generation time \(T\), the reproductive
quantiles \(B_{50}\), \(B_{75}\), \(B_{90}\), and \(B_{95}\), and, when it
exists, the entropy-maximizing endpoint \(D^*\).

\begin{table}[H]
\centering
\caption{Longevity-calibrated demographic scenarios and evolutionary
entropy sensitivity for the Greenland shark. The values of \(q\) are
the midpoints of the intervals for which the \(95\%\) reproductive
quantile \(B_{95}\) matches the lower, central, and upper longevity
estimates of \(272\), \(392\), and \(512\) years, respectively.}
\label{tab:calibrated-sensitivity}
\begin{tabular}{ccccccccccc}
\hline
\(q\)
& \(T\)
& \(B_{50}\)
& \(B_{75}\)
& \(B_{90}\)
& \(B_{95}\)
& \(D^{*}\)
& \(\Delta q=q-q_c\)
& \(H(q)\)
& \(dH/dq\)
& Regime
\\
\hline
0.974613
& 194.39
& 182
& 209
& 245
& 272
& --\(^{\dagger}\)
& -0.001684
& 0.023976
& 0.0139
& Below \(q_c\)
\\
0.987413
& 234.45
& 210
& 265
& 337
& 392
& 299
& 0.011116
& 0.022900
& -0.2755
& Above \(q_c\)
\\
0.991632
& 274.50
& 238
& 320
& 430
& 512
& 272
& 0.015335
& 0.021053
& -0.6581
& Above \(q_c\)
\\
\hline
\end{tabular}
\end{table}

{\footnotesize
$^\dagger$ 
For $q < q_c$ the  evolutionary entropy increases monotonically to its asymptotic value, admitting no finite maximizer.
}
\vspace{5pt}

The three calibrated scenarios produce markedly different reproductive
time scales. By construction, the \(95\%\) reproductive quantiles span the
full reported longevity interval from \(272\) to \(512\) years. The lower
longevity estimate corresponds to a value of \(q\) slightly below the
critical threshold, whereas the central and upper estimates correspond to
supercritical values. Thus, under the \(B_{95}\)-based calibration adopted here, the values of
\(q\) induced by the reported longevity interval straddle the critical
boundary separating asymptotic and finite entropy regimes.

To assess the sensitivity to the choice of calibration quantile, we fix the target age at \(392\) years and assign it successively
to \(B_{90}\), \(B_{95}\), and \(B_{99}\). The corresponding values of \(q\), obtained as the
midpoint of the interval in \eqref{eq:quantile-interval}, are shown
below.

\begin{table}[H]
\centering
\caption{Sensitivity to the quantile assigned to the target age
\(392\), with \(A=156\).}
\vspace{1mm}
\label{tab:quantile-sensitivity}
\begin{tabular}{cccc}
\hline
$p$ & $q$ & $q-q_c$ & $T$\\
\hline
0.90 & 0.990311 & 0.014014 & 258.21\\
0.95 & 0.987413 & 0.011116 & 234.45\\
0.99 & 0.980716 & 0.004419 & 206.86\\
\hline
\end{tabular}
\end{table}

The inferred generation time varies by approximately \(51\) years
between the \(B_{90}\) and \(B_{99}\) calibrations. In every case, however,
\(q>q_c\), so the classification remains supercritical.

\section{Variational Interpretation and Entropy Sensitivity}
\label{sec:sensitivity}

\textcolor{black}{Sensitivity analysis of evolutionary entropy in age-structured Leslie
populations was developed by Demetrius, Gundlach and Ziehe
\cite{Demetrius07}. Steiner and Tuljapurkar
\cite{SteinerTuljapurkar2020} subsequently derived exact sensitivities
of population entropy under constrained perturbations of the transition
probabilities of ergodic and absorbing Markov chains. More recently,
Oliveira \cite{Oliveira2026} developed explicit sensitivity
formulae for evolutionary entropy in irreducible Lefkovitch matrices,
including open-group Leslie matrices as a direct specialization.}

\textcolor{black}{In the geometric open-group family considered here, the demographic
variation is restricted to the single persistence parameter \(q\).
Evolutionary entropy is given explicitly by
equation~\eqref{eq:Hq}; differentiating \(H(q)\) and simplifying
therefore yields the sensitivity along this one-parameter family,}
\[
\frac{dH}{dq}
=
\frac{\log(1-q)-A\log q}
     {\bigl[A(1-q)+q\bigr]^2}.
\]
Consequently,
\begin{equation}
\label{eq:sensitivity}
\frac{dH}{dq}=0
\qquad\Longleftrightarrow\qquad
q^A+q-1=0.
\end{equation}
Equation~\eqref{eq:sensitivity} shows that the criticality equation
\eqref{eq:critical} is exactly the stationarity condition for
evolutionary entropy within the geometric open-group family. Since
Theorem 4.10 of \cite{BuescuElaydiOliveira2026} establishes that this
stationary point is the unique maximum, the critical persistence
parameter \(q_c\) is precisely the entropy-maximizing value of \(q\).
Thus the threshold separating the asymptotic and finite-window regimes
coincides with the variational maximum of evolutionary entropy.

For completeness, the corresponding sensitivities of $S$
and $T$ are
\[
\begin{aligned}
\frac{dS}{dq}&=-\frac{\log q}{(1-q)^2},
&
\frac{dT}{dq}&=\frac{1}{(1-q)^2}.
\end{aligned}
\]
Since $0<q<1$, both quantities increase monotonically with $q$.
Therefore the critical behaviour of the evolutionary
entropy is a property of the global quotient $H=S/T$, rather than of its numerator or denominator
taken separately.

Because $dH/dq$ vanishes at the unique entropy-maximizing value $q_c$,
the entropy is locally insensitive to sufficiently small perturbations
of $q$ near the critical threshold. At the critical point,
\[
\frac{d^2H}{dq^2}(q_c)\approx -9.246<0.
\]
As $q$ moves into the calibrated supercritical regime, $dH/dq$
becomes increasingly negative. Thus larger supercritical values of
$q$ produce increasingly steep decreases in entropy, as shown in Table~\ref{tab:calibrated-sensitivity} and Figure~\ref{fig:SHARK_PEAK}.

\FloatBarrier

\subsection{Sensitivity to Age at Maturity}

For $A > 0$, implicit differentiation of
\[
q_c(A)^A+q_c(A)-1=0
\]
gives
\begin{equation}
\frac{dq_c}{dA}
=
-\frac{q_c^A\log q_c}
       {Aq_c^{A-1}+1}
>0.
\label{eq:qc-derivative}
\end{equation}
Thus later maturity shifts the critical threshold towards \(1\). Setting
\(\varepsilon_A=1-q_c(A)\), the criticality equation becomes
\[
(1-\varepsilon_A)^A=\varepsilon_A,
\]
and, for large \(A\),
\begin{equation}
q_c(A)
\simeq
1-\frac{W(A)}{A},
\label{eq:qc-lambert}
\end{equation}
where \(W\) denotes the principal Lambert \(W\)-function. For \(A=156\),
this approximation gives \(q_c\simeq0.976072\),  in excellent agreement with the exact value \(0.976297\).

The reported female sexual-maturity estimate, used here as a proxy
for the age at first reproduction, is $A=156\pm22$ years
\cite{Nielsen2016GreenlandShark}.

We next examine
\textcolor{black}{the sensitivity of the entropy-derived quantities across this interval.}

Given $0 < p < 1$, we denote the reproductive quantile corresponding to $p$ by $B_p$. As shown in \cite{BuescuElaydiOliveira2026}, 
for fixed reproductive persistence $q$ the reproductive quantiles
satisfy 
\[
B_p
=
A+
\left\lceil
\frac{\log(1-p)}{\log q}
\right\rceil
-1 \ \stackrel{\rm def}{\equiv} \  A + C_p(q), 
\]
where $C_p(q)$ depends only on $q$. 
This means that, for fixed reproductive persistence $q$, a change in
the age at first reproduction translates every reproductive quantile
by exactly the same amount, while preserving the shape of the
reproductive window.

Table~\ref{tab:sensitivity} presents the critical threshold and several
entropy-derived quantities across the reported uncertainty interval, using the central longevity-calibrated value
\(q=0.987413\).

\begin{table}[H]
\centering
\caption{\textcolor{black}{Sensitivity of the entropy-derived quantities to the
reported uncertainty in age at maturity, \(A=156\pm22\) years, using the
central longevity-calibrated scenario \(q=0.987413\).}}
\label{tab:sensitivity}
\begin{tabular}{ccccccc}
\hline
$A$ & $q_c$ & $\Delta q$ & $T$ & $B_{90}$ & $B_{95}$ & $B_{99}$\\
\hline
134 & 0.973319 & 0.014094 & 212.447 & 315 & 370 & 497\\
145 & 0.974906 & 0.012507 & 223.447 & 326 & 381 & 508\\
156 & 0.976297 & 0.011116 & 234.447 & 337 & 392 & 519\\
167 & 0.977529 & 0.009884 & 245.447 & 348 & 403 & 530\\
178 & 0.978627 & 0.008786 & 256.447 & 359 & 414 & 541\\
\hline
\end{tabular}
\end{table}

The critical threshold remains extremely close to \(1\) throughout the
entire maturity range, while the distance
\(\Delta q=0.987413-q_c\) varies only slightly. Moreover, the reproductive
quantiles increase exactly in parallel with \(A\), in agreement with
the relation \(B_p=A+C_p(q)\). The results therefore indicate that
uncertainty in the age at first reproduction affects primarily the location
of the reproductive window rather than its overall shape.

\begin{figure}[ht]
	\centering
	\includegraphics[width=0.8\textwidth]{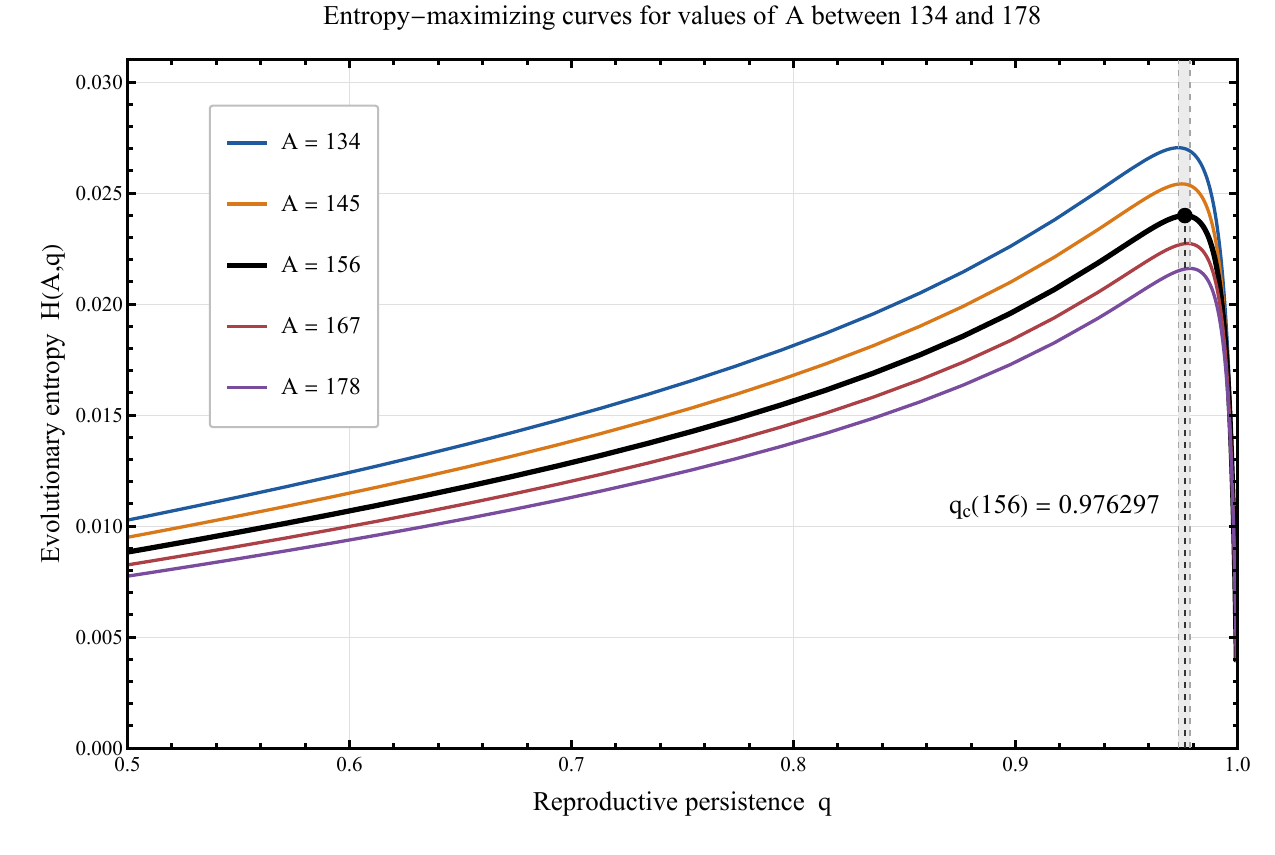}
	\caption{Evolutionary entropy profiles for different ages at first reproduction. The light-grey band marks the narrow interval spanned by the five entropy-maximizing thresholds \(q_c(A)\), while the dashed line identifies \(q_c(156)\).}\label{fig:SHARK_PEAK}
	\end{figure}

Figure~\ref{fig:SHARK_PEAK} illustrates \textcolor{black}{their sensitivity to variations
in the age of sexual maturity.} Here \(A\) varies across the reported
uncertainty interval \(156\pm22\) years, from \(134\) to \(178\) years.
As \(A\) increases, the peak of \(H(q)\) shifts rightward and the curve flattens,
pushing the entropy maximum closer to \(1\).
The Greenland shark thus sits in a particularly flat region of the entropy
landscape, so that evolutionary entropy itself is relatively insensitive to
small perturbations of \(q\) near the critical threshold, even though the
associated reproductive time scales may vary substantially.

\FloatBarrier
{\color{black}
\subsection{Asymptotic Divergence of the Optimal Reproductive Window}

In this section we follow closely the methods and results in \cite{BuescuElaydiOliveira2026}. 
Let \(n=D-A+1\) denote the length of a finite reproductive window,
and let \(x>0\) be its continuous counterpart. With
\(\beta=-\log q\), define the functions 

\[
\Psi(x,q) =
\frac{\beta(1-q^x)}
{\beta x-1+q^x}, \qquad \qquad 
F(x,q)
=
\frac{\displaystyle
	\log\!\left(\frac{1-q^x}{1-q}\right)+A\log q}
{\displaystyle
	A+\frac{q}{1-q}-\frac{xq^x}{1-q^x}}.
\]

The stationarity equation for the continuous truncated entropy is then written  
\[
\Psi(x,q)=F(x,q).
\]
 As \(x\to\infty\), we have 
\(\Psi(x,q)\sim 1/x\), whereas
\begin{equation}
F(x,q)\longrightarrow
F_\infty(A,q)
=
\frac{-\log(1-q)+A\log q}
     {A+q/(1-q)}.
\label{eq:Finfty}
\end{equation}

\begin{proposition}[Critical divergence of the optimal window]
For fixed \(A\) and \(q>q_c(A)\), let \(x^*(q)\) be the unique
continuous maximizer of the truncated geometric entropy. Then
\begin{equation}
x^*(q)
\sim
\frac{q_c}{q-q_c},
\qquad
q\downarrow q_c.
\label{eq:critical-divergence}
\end{equation}
The integer maximizer \(n^*\) has the same leading behaviour.
\end{proposition}

\begin{proof}
Set
\[
N(q)=-\log(1-q)+A\log q,
\qquad
D(q)=A+\frac{q}{1-q}.
\]
Since \(N(q_c)=0\),
\[
F_\infty(A,q)
=
\frac{N'(q_c)}{D(q_c)}(q-q_c)
+
O\!\left((q-q_c)^2\right).
\]
Using ${\displaystyle \frac{N'(q_c)}{D(q_c)}
	=
	\frac{1}{q_c},}
$ we obtain

\begin{equation}
F_\infty(A,q)
=
\frac{q-q_c}{q_c}
+
O\!\left((q-q_c)^2\right).
\label{eq:Finfty-expansion}
\end{equation}
As \(q\downarrow q_c\),
\(x^*(q)\to\infty\). Taking into account that 
\(\Psi(x^*,q)\sim1/x^*\) yields
\eqref{eq:critical-divergence}.
\end{proof}

 Define the normalized distance from the critical threshold by
${\displaystyle \eta
	=
	\frac{q-q_c(A)}{1-q_c(A)}. }$
Then \eqref{eq:critical-divergence} becomes
\begin{equation}
x^*(q)
\sim
\frac{\Gamma(A)}{\eta},
\qquad
\Gamma(A)
=
\frac{q_c(A)}{1-q_c(A)}.
\label{eq:amplification-factor}
\end{equation}
The same asymptotic relation holds for the integer optimizer \(n^*\).
For the Greenland shark,
$
\Gamma(156)=41.19,
$
so a small normalized distance above the threshold corresponds to a
much larger reproductive-window length.

}

\section{Comparative Context Among Exceptionally Long-Lived Species}
\label{app:comparative_context}

To place the Greenland shark in a broader demographic context, we compare
it with five other exceptionally long-lived species at the same relative
distance above their respective critical thresholds.
Taking as reference the Greenland shark scenario defined by the midpoint
\(q_{\rm GS}=0.987412888\ldots\) of the interval for which
\(B_{95}=392\) years, the central longevity estimate, we define the normalized supercritical parameter
\[
\eta^{\rm GS}=
\frac{q^{\rm GS}-q_c^{\rm GS}}
{1-q_c^{\rm GS}}
=
0.468961\ldots
\]
and, for each species \(i\), the corresponding persistence parameter
\[
q_i=q_{c,i}+\eta^{\rm GS} (1-q_{c,i}).
\]
This provides a common normalized supercritical reference for comparing
the resulting demographic time scales.

\begin{table}[H]
	\centering
	\scriptsize
	\setlength{\tabcolsep}{2.7pt}
	\renewcommand{\arraystretch}{1.08}
	\caption{Entropy-based demographic quantities for six exceptionally
		long-lived species at the common normalized supercritical parameter value \(\eta^{\rm GS}=0.468961\).
		Here \(A\) is the age at first reproduction, \(L\) is the reference
		longevity value used in the comparison,
		\(q_c\) is the critical persistence threshold,
		\(\Delta q=q-q_c\), \(T\) is generation time, \(H(q)\) is evolutionary
		entropy, \(B_p\) are the discrete reproductive quantile ages, and \(D^*\)
		is the entropy-maximizing finite reproductive endpoint. Values of \(q\)
		are displayed to six decimal places; all derived quantities were computed
		before rounding.}
	\label{tab:comparative_extreme}
	\resizebox{\textwidth}{!}{%
		\begin{tabular}{lrrrrrrrrrrrrr}
			\toprule
			Species
			& \(A\)
			& \(L\)
			& \(q_c\)
			& \(q\)
			& \(\Delta q\)
			& \(T\)
			& \(H(q)\)
			& \(B_{50}\)
			& \(B_{75}\)
			& \(B_{90}\)
			& \(B_{95}\)
			& \(B_{99}\)
			& \(D^*\) \\
			\midrule
			
			\textit{Somniosus microcephalus}
			\cite{Nielsen2016GreenlandShark}
			& 156 & 392 & 0.976297 & 0.987413 & 0.011116
			& 234.45 & 0.022900 & 210 & 265 & 337 & 392 & 519 & 299 \\
			
			\textit{Balaena mysticetus}
			\cite{George1999BowheadWhale}
			& 20 & 211 & 0.893895 & 0.943654 & 0.049759
			& 36.75 & 0.104702 & 31 & 43 & 59 & 71 & 99 & 50 \\
			
			\textit{Hoplostethus atlanticus}
			\cite{FrancisHorn1997OrangeRoughy}
			& 30 & 250 & 0.919461 & 0.957231 & 0.037769
			& 52.38 & 0.078850 & 45 & 61 & 82 & 98 & 135 & 70 \\
			
			\textit{Arctica islandica}
			\cite{Butler2013Arctica,RidgwayRichardsonAustad2011,RipleyCaswell2008}
			& 55 & 507 & 0.947748 & 0.972252 & 0.024504
			& 90.04 & 0.050762 & 79 & 104 & 136 & 161 & 218 & 119 \\
			
			\textit{Aldabrachelys gigantea}
			\cite{Bourn1977Aldabra,Gerlach2004GiantTortoises}
			& 25 & 255 & 0.908697 & 0.951515 & 0.042817
			& 44.62 & 0.089678 & 38 & 52 & 71 & 85 & 117 & 60 \\
			
			\textit{Phoebastria immutabilis}
			\cite{Finkelstein2010LaysanAlbatross,USFWS2025Wisdom}
			& 8 & 75 & 0.811652 & 0.899980 & 0.088328
			& 17.00 & 0.191235 & 14 & 21 & 29 & 36 & 51 & 24 \\
			
			\bottomrule
		\end{tabular}%
	}
\end{table}
\FloatBarrier

\textcolor{black}{All six supercritical scenarios satisfy $\Phi<0$ as a consequence
of equation~ \eqref{eq:super}.}

Table~\ref{tab:comparative_extreme} isolates the exceptional demographic
time scale of the Greenland shark. The comparison is particularly revealing
because the reported maximum longevity of \textit{Arctica islandica}
(\(507\) years) exceeds the central age estimate used here for the oldest
Greenland shark (\(392\) years), yet its age at first reproduction and
critical threshold remain far below those of \textit{S. microcephalus}.
At the same normalized position above the critical threshold, the Greenland
shark has a generation time of approximately \(234\) years, compared with
\(90\) years for \textit{A. islandica}, the next largest value in the
comparison. Its \(99\%\) reproductive quantile reaches \(519\) years,
whereas the corresponding value for \textit{A. islandica} is \(218\) years.
Moreover, by construction of the reference scenario, the \(95\%\)
reproductive quantile of the Greenland shark coincides with the central
reported longevity estimate, \(B_{95}=392\) years.

The finite entropy-selected endpoint shows the same distinction. For the
Greenland shark, \(D^*=299\) years, approximately \(76\%\) of the central
longevity estimate used as the reference value in the comparison. For the
remaining species, the corresponding proportion ranges only from
approximately \(23\%\) to \(32\%\). This comparison should not be
interpreted as a fitted demographic reconstruction. Rather, it shows that,
under a common supercritical reference, the Greenland shark alone combines
extremely delayed maturity, a critical threshold exceptionally close to
unity, and reproductive time scales extending over several centuries.


\FloatBarrier

\subsection{The bowhead whale as an empirical benchmark}
\label{subsec:bowhead-benchmark}

The bowhead whale, \textit{Balaena mysticetus}, provides an
independent demographic benchmark because Brandon and Wade
\cite{BrandonWade2006} estimated an age- and sex-structured Leslie
model directly from population data. Their model specifies juvenile
and adult survival, age at maturity, fecundity, and a finite maximum
represented age. Figures~\ref{fig:bowhead-life-history} and
\ref{fig:bowhead-leslie-graph} summarize, respectively, the
corresponding life-history schedules and the structure of the
finite Leslie model.

\FloatBarrier

\begin{figure}[H]
	\centering
	\includegraphics[width=0.9\textwidth]{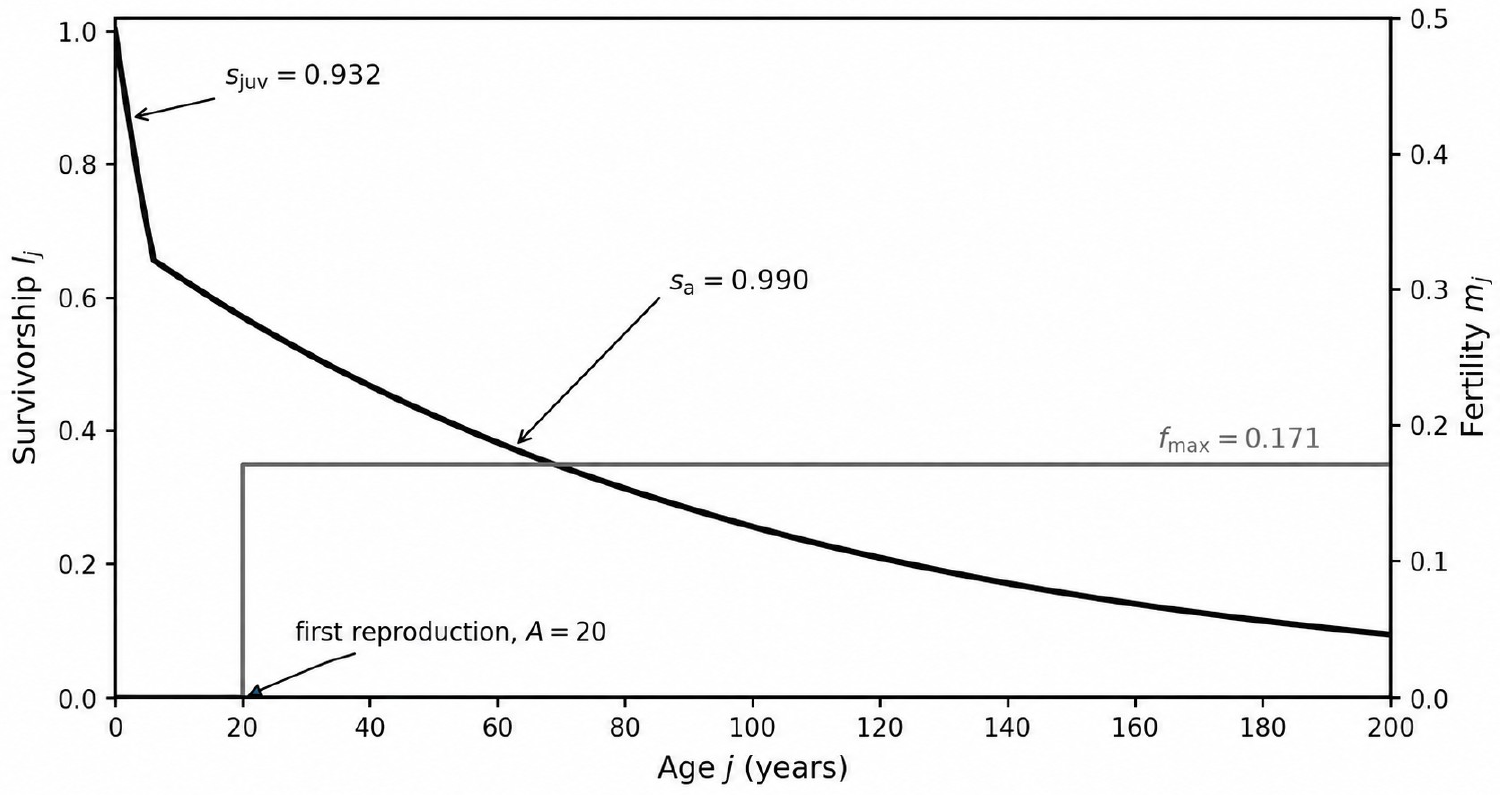}
	\caption{Representative life-history schedules for the
	bowhead-whale benchmark. The survivorship curve $\ell_j$ is
	obtained from the Brandon--Wade model-averaged survival
	estimates $s_{\rm juv}\simeq0.932$ and $s_a\simeq0.990$,
	using the representative central value $a_T=5$. For the
	comparison considered here, fertility is zero before the
	first reproductive age $A=20$ and constant thereafter at
	$f_{\max}\simeq0.171$. The age range extends to
	$a_{\max}=200$, the final age represented in the original
	Brandon--Wade Leslie model.}
	\label{fig:bowhead-life-history}
\end{figure}

Brandon and Wade \cite{BrandonWade2006} parameterize their model by
the juvenile annual survival probability $s_{\rm juv}$, the last age
$a_T$ subject to juvenile survival, the adult annual survival
probability $s_a$, the age-at-maturity parameter $a_m$, the maximum
fertility $f_{\max}$, and the maximum represented age $a_{\max}$.
In their convention, $a_m$ is the last age class with zero
fecundity; hence the first reproductive age in the notation used
throughout the present paper is
\[
A=a_m+1.
\]
Their Bayesian model-averaged median estimates include
\[
s_a\simeq0.990,\qquad
s_{\rm juv}\simeq0.932,\qquad
f_{\max}\simeq0.171,
\]
together with a maximum annual population increase of approximately
$4.3\%$ \cite{BrandonWade2006}. Their analysis supports maximum
population increases between approximately $3\%$ and $5\%$.

We use the standard discrete age-structured notation for Leslie
models \cite{ElaydiCushing2024,Key}. Thus $m_j$ denotes the fertility
contribution of an individual of age $j$, $s_j$ denotes the
conditional probability of surviving from age $j$ to age $j+1$,
and $\lambda$ denotes the dominant eigenvalue of the annual
projection matrix. Let $\ell_j$ denote survivorship from birth to
age $j$, with
\begin{equation}
\ell_0=1,\qquad
\ell_j=\prod_{k=0}^{j-1}s_k,\qquad j\geq1.
\label{eq:bowhead-survivorship-definition}
\end{equation}

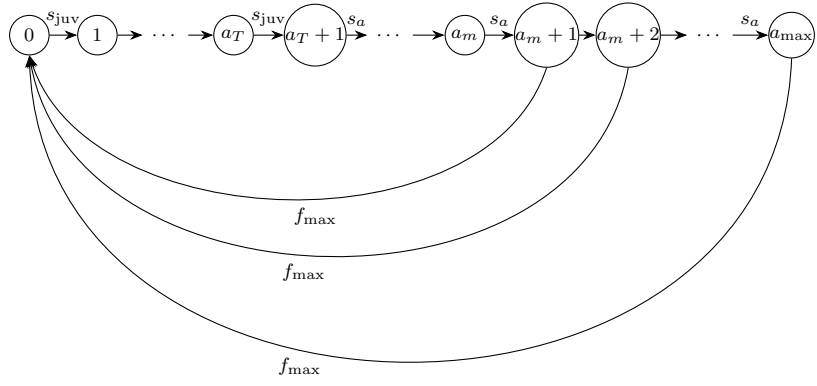
\begin{figure}[ht]
\centering

\begin{tikzpicture}[
    scale=0.9,
    transform shape,
    >=Stealth,
    every node/.style={font=\scriptsize},
    state/.style={
        circle,
        draw=black,
        minimum size=5.8mm,
        inner sep=0pt
    }
]


\node[state] (x0)    at (0,0)    {$0$};
\node[state] (x1)    at (1.0,0)  {$1$};
\node        (d1)    at (2.0,0)  {$\cdots$};
\node[state] (xT)    at (3.0,0)  {$a_T$};
\node[state] (xTp1)  at (4.2,0)  {$a_T+1$};
\node        (d2)    at (5.3,0)  {$\cdots$};

\node[state] (xm)    at (6.4,0)  {$a_m$};
\node[state] (xmp1)  at (7.6,0)  {$a_m+1$};
\node[state] (xmp2)  at (8.8,0)  {$a_m+2$};
\node        (d3)    at (10.0,0) {$\cdots$};
\node[state] (xmax)  at (11.2,0) {$a_{\max}$};


\draw[->] (x0) -- node[above] {$s_{\rm juv}$} (x1);
\draw[->] (x1) -- (d1);
\draw[->] (d1) -- (xT);
\draw[->] (xT) -- node[above] {$s_{\rm juv}$} (xTp1);

\draw[->] (xTp1) -- node[above] {$s_a$} (d2);
\draw[->] (d2) -- (xm);
\draw[->] (xm) -- node[above] {$s_a$} (xmp1);
\draw[->] (xmp1) -- (xmp2);
\draw[->] (xmp2) -- (d3);
\draw[->] (d3) -- node[above] {$s_a$} (xmax);


\draw[->]
    (xmp1.south)
    to[out=-108,in=-72,looseness=0.95]
    node[pos=.46,below] {$f_{\max}$}
    (x0.south);

\draw[->]
    (xmp2.south)
    to[out=-100,in=-80,looseness=1.12]
    node[pos=.53,below] {$f_{\max}$}
    (x0.south);

\draw[->]
    (xmax.south)
    to[out=-92,in=-88,looseness=1.36]
    node[pos=.60,below] {$f_{\max}$}
    (x0.south);

\end{tikzpicture}

\caption{Life-cycle graph of the finite age-structured Leslie model
for \textit{Balaena mysticetus} used by Brandon and Wade
\cite{BrandonWade2006}. Juvenile annual survival is $s_{\rm juv}$
through age $a_T$ and adult annual survival is $s_a$ thereafter.
The class $a_m$ is the last age class with zero fecundity, so
fertility begins at age $a_m+1$ and continues through the final
represented age $a_{\max}=200$. The three return arcs represent the
fertility contributions from the first, second, and final
reproductive age classes.}
\label{fig:bowhead-leslie-graph}

\end{figure}

This is the standard life-table survivorship function in discrete
age \cite{KeyfitzCaswell2005}.

Brandon and Wade \cite{BrandonWade2006} denote their maximum finite
annual proportional increase by $r_{\max}$, so that the corresponding
dominant eigenvalue is $1+r_{\max}$. To distinguish this convention
from the Malthusian parameter $r=\log\lambda$ used throughout the
present paper, we write this proportional increase as $g_{\rm BW}$.
Thus
\begin{equation}
\lambda=1+g_{\rm BW},
\qquad
g_{\rm BW}\simeq0.043,
\qquad
\lambda\simeq1.043.
\label{eq:bowhead-lambda}
\end{equation}

For the entropy-based comparison we take the first reproductive age
to be $A=20$, as in Table~\ref{tab:comparative_extreme}. In the
notation of the Brandon--Wade graph this corresponds to the first
fertile class $a_m+1$. Their original Leslie model is truncated at
$a_{\max}=200$ years, after which survival is zero
\cite{BrandonWade2006}; this finite structure is shown in
Figure~\ref{fig:bowhead-leslie-graph}. For the analytical comparison
developed here, we replace this terminal truncation by the open
post-maturity representation of
Section~\ref{Open-Group Representation}, retaining the reported
adult survival $s_a$. Consistently with the Brandon--Wade assumption
of identical fertility for all mature animals, fertility is taken
to be constant after maturity. The resulting fertility and
survivorship schedules are summarized in
Table~\ref{tab:bowhead-vital-schedules}.

\begin{table}[ht]
\centering
\small
\caption{Age-specific fertility $m_j$ and survivorship $\ell_j$ in
the open-terminal extension used in the present analysis, constructed
from the Brandon--Wade bowhead demographic parametrization
\cite{BrandonWade2006}. The notation follows the standard discrete
Leslie and life-table formulation
\cite{ElaydiCushing2024,KeyfitzCaswell2005}.}
\label{tab:bowhead-vital-schedules}
\begin{tabular}{c|c|c|c|c}
\hline
Age $j$
& $0$
& $1\leq j\leq a_T+1$
& $a_T+2\leq j<A$
& $j\geq A$ \\
\hline
$m_j$
& $0$
& $0$
& $0$
& $f_{\max}=0.171$ \\[1mm]
$\ell_j$
& $1$
& $s_{\rm juv}^{\,j}$
& $s_{\rm juv}^{\,a_T+1}s_a^{\,j-a_T-1}$
& $s_{\rm juv}^{\,a_T+1}s_a^{\,j-a_T-1}$ \\
\hline
\end{tabular}
\end{table}

Equivalently, the annual survival schedule is
\begin{equation}
s_j=
\begin{cases}
s_{\rm juv}, & 0\leq j\leq a_T,\\
s_a,         & j>a_T,
\end{cases}
\label{eq:bowhead-survival-schedule}
\end{equation}
which, together with
Eq.~\eqref{eq:bowhead-survivorship-definition}, completely specifies
the open life table. Brandon and Wade allow $a_T$ to range from
$1$ to $9$ years \cite{BrandonWade2006}; for the representative
central life table below we take $a_T=5$.

Since both survival and fertility are constant after maturity, the
reproductive tail has the geometric open-group form considered in
Section~\ref{Open-Group Representation}. Its persistence parameter is
therefore
\begin{equation}
\rho=s_a,\qquad
q=\frac{\rho}{\lambda}
  =\frac{0.990}{1.043}
  \simeq0.94919.
\label{eq:bowhead-q}
\end{equation}
Using the definitions of generation time, evolutionary entropy and
reproductive potential introduced above, the corresponding
macroscopic demographic parameters are
\begin{equation}
T\simeq38.68,\qquad
H\simeq0.10222,\qquad
r=\log\lambda\simeq0.04210,\qquad
\Phi=r-H\simeq-0.06012.
\label{eq:bowhead-macroscopic}
\end{equation}
The critical persistence associated with $A=20$ is
\begin{equation}
q_c(20)\simeq0.89390<q,
\label{eq:bowhead-critical}
\end{equation}
so the independently parametrized bowhead model lies in the
supercritical regime. Its reproductive quantiles are
\begin{equation}
(B_{50},B_{75},B_{90},B_{95},B_{99})
=(33,46,64,77,108)\ {\rm years}.
\label{eq:bowhead-quantiles}
\end{equation}

The survival schedule also permits an independent characterization
of the longevity component of the life history. In standard
life-table notation, life expectancy at birth is the total area
under the survivorship curve; in annual discrete time this becomes
\begin{equation}
e_0=\sum_{j=0}^{\infty}\ell_j
\label{eq:bowhead-life-expectancy}
\end{equation}
\cite{KeyfitzCaswell2005}. Following the entropic formulation of
longevity developed by Demetrius and, in particular, the definitions
given by Demetrius, Sahasranaman and Ziehe
\cite{Demetrius1974,DemetriusSahasranamanZiehe2024}, define the
normalized survivorship distribution
\begin{equation}
p_j^\ast=\frac{\ell_j}{e_0},
\qquad
\sum_{j=0}^{\infty}p_j^\ast=1.
\label{eq:bowhead-normalized-survivorship}
\end{equation}
The discrete counterparts of the life-table entropy and life-span
potential are then
\begin{equation}
H^\ast
=
-\sum_{j=0}^{\infty}
p_j^\ast\log p_j^\ast
=
-\sum_{j=0}^{\infty}
\frac{\ell_j}{e_0}
\log\frac{\ell_j}{e_0},
\label{eq:bowhead-lifetable-entropy}
\end{equation}
and
\begin{equation}
\Phi^\ast
=
-\sum_{j=0}^{\infty}
p_j^\ast\log\ell_j
=
-\sum_{j=0}^{\infty}
\frac{\ell_j}{e_0}\log\ell_j,
\label{eq:bowhead-lifespan-potential}
\end{equation}
respectively
\cite{DemetriusSahasranamanZiehe2024}. They satisfy the exact
identity
\begin{equation}
\log e_0=H^\ast-\Phi^\ast,
\label{eq:bowhead-lifetable-identity}
\end{equation}
which is the survivorship analogue of the decomposition
$r=H+\Phi$
\cite{Demetrius1974,DemetriusSahasranamanZiehe2024}.

Table~\ref{tab:bowhead-benchmark} summarizes the entropy-based and
life-table quantities obtained from the reported model-averaged
demographic estimates. The reproductive ranges show the variation
induced by the Brandon--Wade interval
$g_{\rm BW}=0.03$--$0.05$, holding $A$ and $s_a$ fixed. The
life-table ranges show the variation induced by
$a_T=1,\ldots,9$, holding the model-averaged survival estimates
fixed. They are therefore sensitivity ranges rather than joint
posterior credible intervals.

\begin{table}[ht]
\centering
\caption{Entropy-based and life-table demographic quantities for
\textit{Balaena mysticetus} derived from the Brandon--Wade
model-averaged demographic estimates \cite{BrandonWade2006}.
Reproductive central values use $A=20$, $s_a=0.990$ and
$\lambda=1.043$. The reproductive sensitivity range corresponds to
$g_{\rm BW}=0.03$--$0.05$. For the life-table quantities the central
value uses $a_T=5$, and the sensitivity range spans
$a_T=1,\ldots,9$, with $s_{\rm juv}=0.932$ and $s_a=0.990$.}
\label{tab:bowhead-benchmark}
\begin{tabular}{lcc}
\hline
Quantity & Central value & Sensitivity range \\
\hline
$\lambda$
    & $1.043$    & $1.030$--$1.050$ \\
$q=\rho/\lambda$
    & $0.94919$  & $0.94286$--$0.96117$ \\
$q_c$
    & $0.89390$  & -- \\
$T$ (years)
    & $38.68$    & $36.50$--$44.75$ \\
$H$
    & $0.10222$  & $0.09450$--$0.10502$ \\
$r=\log\lambda$
    & $0.04210$  & $0.02956$--$0.04879$ \\
$\Phi$
    & $-0.06012$ & $-0.06494$--$-0.05623$ \\
$B_{50}$ (years)
    & $33$       & $31$--$37$ \\
$B_{75}$ (years)
    & $46$       & $43$--$54$ \\
$B_{90}$ (years)
    & $64$       & $59$--$78$ \\
$B_{95}$ (years)
    & $77$       & $70$--$95$ \\
$B_{99}$ (years)
    & $108$      & $98$--$136$ \\
$e_0$ (years)
    & $70.61$    & $56.88$--$88.79$ \\
$H^\ast$
    & $5.5845$   & $5.5543$--$5.5982$ \\
$\Phi^\ast$
    & $1.3274$   & $1.1119$--$1.5133$ \\
\hline
\end{tabular}
\end{table}

The independently derived reproductive quantities in
Table~\ref{tab:bowhead-benchmark} are close to those obtained for
the bowhead whale from the normalized comparison in
Table~\ref{tab:comparative_extreme}. More importantly, the published
demographic parametrization permits both the reproductive pair
$(H,\Phi)$ and the survivorship pair $(H^\ast,\Phi^\ast)$ to be
evaluated directly. The bowhead whale therefore provides an
empirical contrast with the Greenland shark, for which the available
observations constrain the temporal organization of reproduction
without determining a complete Leslie matrix.

\section{Conclusion}

The Greenland shark provides an extreme application of the critical-threshold
theory developed in \cite{BuescuElaydiOliveira2026}. With female maturity
estimated at approximately \(156\) years, the corresponding critical
threshold \(q_c\approx0.9763\) lies remarkably close to 1. The
longevity-calibrated scenarios considered here span this critical boundary,
linking the reported longevity range to both asymptotic and finite-window
entropy regimes.

These results also reveal a different form of sensitivity. Because the age
at first reproduction is exceptionally large, the longevity-calibrated
scenarios considered here lie close to, and straddle, the demographic
critical boundary. Small changes in the effective reproductive ratio
\(q=\rho/\lambda\) may therefore move the model between asymptotic and
finite-window entropy regimes and substantially alter the structure of the
entropy-selected reproductive window. The corresponding reproductive
quantiles span the reported longevity interval of the Greenland shark,
while the finite entropy-maximizing endpoint changes sharply and disappears
altogether on the subcritical side.
The asymptotic relation \eqref{eq:critical-divergence} quantifies this
sensitivity. The normalized amplification factor 
\(\Gamma(156)=41.19\) shows that small changes in reproductive persistence
produce much larger changes in the entropy-selected window.

The bowhead whale provides an independent empirical benchmark for this
construction. Its published Leslie demography leads to entropy-based
reproductive quantities close to those obtained from the reduced
open-group comparison, while also permitting direct calculation of
$r$, $\Phi$, and the survivorship quantities $e_0$, $H^\ast$, and
$\Phi^\ast$. The contrast is informative: for the bowhead whale the
demographic matrix is available, whereas for the Greenland shark the
same theory is used in the inverse direction, to constrain the
reproductive organization from the limited demographic information
that is known.

An interesting possible application of our methods should be emphasized. The critical-threshold theory of \cite{BuescuElaydiOliveira2026} converts two empirically accessible quantities, the age at first reproduction and an estimate of longevity, into quantitative constraints on the entire reproductive window. Even when \(\lambda\) and \(\Phi\) are not independently known, evolutionary entropy alone can still yield biologically meaningful constraints on the temporal organization of reproduction. This may open a route to demographic inference for species where direct reconstruction is out of reach and detailed demographic data are absent. 

%
%
%

\section*{Declaration of competing interests}

The authors declare no competing interests.

\section*{Data and code availability}

No new datasets were generated in this study. The empirical values used
in the comparative analysis were obtained from the cited literature.
All numerical calculations are direct evaluations of the formulas given
in the paper and can be readily reproduced. Wolfram Mathematica was used
only for numerical root finding and for producing Figure~\ref{fig:SHARK_PEAK}. No dedicated
code or software package was developed for this study.

\section*{Funding}

The first author acknowledges partial support from Fundação para a
Ciência e a Tecnologia through project \begin{itemize}
\item UID/04561/2025 (DOI: 10.54499/UID/04561/2025).
\end{itemize}

The third author acknowledges partial support from Fundação para a
Ciência e a Tecnologia through CAMGSD -- Centro de Análise Matemática,
Geometria e Sistemas Dinâmicos, under projects
\begin{itemize}
\item UID/04459/2025
(DOI: 10.54499/UID/04459/2025).
\item UID/PRR/04459/2025 
(DOI: 10.54499/UID/PRR/04459/2025).
\end{itemize}

\section*{Author contributions}

All authors contributed equally.

\bibliographystyle{abbrv}
\bibliography{BibloH2}

@article{MacNeilEtAl2012,
  author  = {MacNeil, M. Aaron and McMeans, Bailey C. and
             Hussey, Nigel E. and Vecsei, Paul and Svavarsson, Jorundur and
             Kovacs, Kit M. and Lydersen, Christian and Treble, Margaret A. and
             Skomal, Gregory B. and Ramsey, Michael and Fisk, Aaron T.},
  title   = {Biology of the {Greenland} Shark
             \textit{Somniosus microcephalus}},
  journal = {Journal of Fish Biology},
  year    = {2012},
  volume  = {80},
  number  = {5},
  pages   = {991--1018},
  doi     = {10.1111/j.1095-8649.2012.03257.x}
}

@article{GrantEtAl2018,
  author  = {Grant, Scott M. and Sullivan, Rennie and Hedges, Kevin J.},
  title   = {{Greenland} Shark (\textit{Somniosus microcephalus})
             Feeding Behavior on Static Fishing Gear, Effect of
             {SMART} Hook Deterrent Technology, and Factors Influencing
             Entanglement in Bottom Longlines},
  journal = {PeerJ},
  year    = {2018},
  volume  = {6},
  pages   = {e4751},
  doi     = {10.7717/peerj.4751}
}

@book{Fisher1930,
  author    = {Fisher, Ronald Aylmer},
  title     = {The Genetical Theory of Natural Selection},
  year      = {1930},
  publisher = {Clarendon Press},
  address   = {Oxford}
}

@book{Roff1992,
  author    = {Roff, Derek A.},
  title     = {The Evolution of Life Histories: Theory and Analysis},
  year      = {1992},
  publisher = {Chapman \& Hall},
  address   = {New York}
}

@article{Butler2013Arctica,
  author  = {Butler, P. G. and Wanamaker, A. D. and Scourse, J. D. and
             Richardson, C. A. and Reynolds, D. J.},
  title   = {Variability of marine climate on the {North Icelandic Shelf}
             in a 1357-year proxy archive based on growth increments in
             the bivalve {Arctica islandica}},
  journal = {Palaeogeography, Palaeoclimatology, Palaeoecology},
  volume  = {373},
  pages   = {141--151},
  year    = {2013},
  doi     = {10.1016/j.palaeo.2012.01.016}
}

@book{Charnov1993,
  author    = {Charnov, Eric L.},
  title     = {Life History Invariants: Some Explorations of Symmetry in Evolutionary Ecology},
  publisher = {Oxford University Press},
  address   = {Oxford},
  year      = {1993}
}

@article{Demetrius1974,
  author  = {Demetrius, L.},
  title   = {Demographic Parameters and Natural Selection},
  journal = {Proceedings of the National Academy of Sciences of the United States of America},
  volume  = {71},
  number  = {12},
  pages   = {4645--4647},
  year    = {1974},
  doi     = {10.1073/pnas.71.12.4645}
}

@article{Demetrius1977,
  author  = {Demetrius, L.},
  title   = {Measures of Fitness and Demographic Stability},
  journal = {Proceedings of the National Academy of Sciences of the United States of America},
  volume  = {74},
  number  = {1},
  pages   = {384--386},
  year    = {1977},
  doi     = {10.1073/pnas.74.1.384}
}

@misc{USFWS2025Wisdom,
  author       = {Trott, Matt},
  title        = {{Midway Atoll National Wildlife Refuge:
                   ``A Great Monument to the People Who Passed Here''
                   and a Haven for Native Wildlife}},
  year         = {2026},
  month        = apr,
  day          = {20},
  organization = {{U.S. Fish and Wildlife Service}},
  url          = {https://www.fws.gov/story/2026-04/midway-atoll-national-wildlife-refuge-great-monument-people-who-passed-here},
  urldate      = {2026-08-04},
  note         = {Wisdom, a Laysan albatross, is reported to be
                  approximately 75 years old}
}

@article{Gravel2024,
  author  = {Gravel, Sarah and Bigman, Jennifer S. and Pardo, Sebasti{\'a}n A.
             and Wong, Serena and Dulvy, Nicholas K.},
  title   = {Metabolism, Population Growth, and the Fast--Slow Life History
             Continuum of Marine Fishes},
  journal = {Fish and Fisheries},
  volume  = {25},
  number  = {2},
  pages   = {349--361},
  year    = {2024},
  doi     = {10.1111/faf.12811}
}

@article{Pardo2016,
  author  = {Pardo, Sebasti{\'a}n A. and Kindsvater, Holly K.
             and Reynolds, John D. and Dulvy, Nicholas K.},
  title   = {Maximum Intrinsic Rate of Population Increase in Sharks,
             Rays, and Chimaeras: The Importance of Survival to Maturity},
  journal = {Canadian Journal of Fisheries and Aquatic Sciences},
  volume  = {73},
  number  = {8},
  pages   = {1159--1163},
  year    = {2016},
  doi     = {10.1139/cjfas-2016-0069}
}

@article{A,
  author  = {D. J. Allwright}, 
  title   = {Hypergraphic functions and bifurcations in recurrence relations},
  journal = {Siam Journal on Applied mathematics},
  year    = 1978,
  number  = 4,
  pages   = {687-691},
volume  = 34,
}

@book{And,
  author    = {Aleksander A. Andronov and Aleksandr A. Vitt and Semen E. Khaikin}, 
  title     = {Theory of Oscillators},
  publisher = {Pergammon Press},
  year      = {1959/1963/1966},
  address   = {Oxford, New York},
}

@article{BrandonWade2006,
  author  = {Brandon, J. R. and Wade, P. R.},
  title   = {Assessment of the {Bering--Chukchi--Beaufort Seas} stock of bowhead whales using Bayesian model averaging},
  journal = {Journal of Cetacean Research and Management},
  volume  = {8},
  number  = {3},
  pages   = {225--239},
  year    = {2006},
  doi     = {10.47536/jcrm.v8i3.718}
}

@misc{BuescuElaydiOliveira2026,
  author        = {Buescu, Jorge and Elaydi, Saber N. and Oliveira, Henrique M.},
  title         = {Evolutionary Entropy Shapes Reproductive Lifespan in
                   Age-Structured Populations},
  year          = {2026},
  eprint        = {2606.22001},
  archivePrefix = {arXiv},
  primaryClass  = {q-bio.PE},
  doi           = {10.48550/arXiv.2606.22001},
    note   = {arXiv:2606.22001. DOI: 10.48550/arXiv.2606.22001}
}

@book{Charlesworth,
  author    = {Brian Charlesworth}, 
  title     = {Evolution in Age-Structured Populations},
  publisher = {Cambridge University Press},
  year      = {1994},
  address   = {Cambridge},
  volume={2},  
}

@article{SalgueroGomez2016COMADRE,
  author    = {Salguero-G{\'o}mez, Roberto and Jones, Owen R. and Archer, C. Ruth and Bein, Cyrus and de Buhr, Hauke and Farack, Claudia and Gottschalk, Florian and Hartmann, Anja and Henning, Anne and Hoppe, Gabriel and R{\"o}mer, Gesa and Runge, Jens and Ruoff, Thomas and Sommer, Verena and Wille, Jutta and Voigt, Juliane and Vieregg, Daniel and Buckley, Yohay Carmel and Che-Castaldo, Judy and Hodgson, David and Scheuerlein, Alexander and Caswell, Hal and Vaupel, James W. and Baudisch, Annette},
  title     = {{COMADRE}: A Global Data Base of Animal Demography},
  journal   = {Journal of Animal Ecology},
  year      = {2016},
  volume    = {85},
  number    = {2},
  pages     = {371--384},
  doi       = {10.1111/1365-2656.12482}
}

@article{Demetrius97,
author = {Lloyd  Demetrius}, 
title = {Directionality principles in thermodynamics {and} evolution},
volume = {94}, 
number = {8}, 
pages = {3491-3498}, 
year = {1997}, 
URL = {http://www.pnas.org/content/94/8/3491.abstract}, 
eprint = {http://www.pnas.org/content/94/8/3491.full.pdf}, 
journal = {Proceedings of the National Academy of Sciences} 
}

@article{Demetrius2000,
  author  = {Demetrius, Lloyd},
  title   = {Thermodynamics and Evolution},
  journal = {Journal of Theoretical Biology},
  volume  = {206},
  number  = {1},
  pages   = {1--16},
  year    = {2000},
  doi     = {10.1006/jtbi.2000.2106}
}

@article{Demetrius04,
  author  = {Lloyd  Demetrius and Volker Matthias Gundlach and Gunter Ochs}, 
  title   = {Complexity and demographic stability in population models},
  journal = {Theoretical Population Biology},
  year    =2004 ,
  number  ={3},
  pages   = {211--225},
  volume   =65 ,
}

@article{Demetrius07,
  author  = {Lloyd  Demetrius and Volker Matthias Gundlach and Martin Ziehe}, 
  title   = {Darwinian fitness and the intensity of natural selection: studies in sensitivity analysis},
  journal = {Journal of Theoretical Biology},
  year    =2007 ,
  number  ={4},
  pages   = {641--653},
  volume   =249 ,
}

@article{Demetrius09,
  title={Invasion exponents in biological networks},
  author={Demetrius, Lloyd and Gundlach, Volker Matthias and Ochs, Gunter},
  journal={Physica A: Statistical Mechanics and its Applications},
  volume={388},
  number={5},
  pages={651--672},
  year={2009},
  publisher={Elsevier}
}

@article{DemetriusLegendreHarremoes2009,
  author  = {Demetrius, Lloyd and Legendre, St{\'e}phane and Harrem{\"o}es, Peter},
  title   = {Evolutionary Entropy: A Predictor of Body Size, Metabolic Rate and Maximal Life Span},
  journal = {Bulletin of Mathematical Biology},
  volume  = {71},
  pages   = {800--818},
  year    = {2009},
  doi     = {10.1007/s11538-008-9382-6}
}

@article{DemetriusGundlach2014,
  author  = {Demetrius, Lloyd A. and Gundlach, Volker Matthias},
  title   = {Directionality Theory and the Entropic Principle of Natural Selection},
  journal = {Entropy},
  year    = {2014},
  volume  = {16},
  number  = {10},
  pages   = {5428--5522},
  doi     = {10.3390/e16105428}
}

@article{DemetriusSahasranamanZiehe2024,
  author  = {Demetrius, Lloyd A. and Sahasranaman, Anand and Ziehe, Martin},
  title   = {Directionality theory and mortality patterns across the primate lineage},
  journal = {Biogerontology},
  volume  = {25},
  pages   = {1215--1237},
  year    = {2024},
  doi     = {10.1007/s10522-024-10134-6}
}

@article{Demetrius2023SelfOrganization,
  author        = {Demetrius, Lloyd A.},
  title         = {Self-Organization, Evolutionary Entropy and
                   Directionality Theory},
  year          = {2023},
  eprint        = {2304.14877},
  archivePrefix = {arXiv},
  primaryClass  = {cond-mat.stat-mech},
  doi           = {10.48550/arXiv.2304.14877},
  note   = {arXiv:2304.14877. DOI: 10.48550/arXiv.2304.14877}
}

@book{ElaydiCushing2024,
  author    = {Saber N. Elaydi and Jim Cushing}, 
  title     = {Discrete Mathematical Models in Population Biology
},
  subtitle={Ecological, Epidemic, and Evolutionary Dynamics},
  publisher = {Springer},
  year      = {2024},
  series={Springer Undergraduate Texts in Mathematics and Technology},
  address   = {New York},
    doi        = {10.1007/978-3-031-64795-6},
  isbn       = {978-3-031-64794-9},
  isbn13     = {978-3-031-64795-6}
}

@book{Key,
  title={Applied Mathematical Demography},
  author={Natham Keyfitz},
  year={1977},
   series =    {Springer Texts in Statistics},
  publisher={Springer-Verlag},
  address={New York},
}

@book{KeyfitzCaswell2005,
  author    = {Keyfitz, Nathan and Caswell, Hal},
  title     = {Applied Mathematical Demography},
  edition   = {3},
  publisher = {Springer},
  address   = {New York},
  year      = {2005},
  series    = {Statistics for Biology and Health},
  doi       = {10.1007/b139042}
}

@book{DS,
  author={Wellington de Melo and Sebastian Strien},
  title    = {One-Dimensional Dynamics}, 
  publisher = {Springer},
  address   = {Berlin, Heildelberg},
  year      = 1993,
}

@article{Oliveira2026,
  author  = {Oliveira, Henrique M.},
  title   = {Sensitivity of Evolutionary Entropy in {Lefkovitch} Matrices},
  journal = {arXiv preprint arXiv:2606.26170},
  year    = {2026},
  eprint  = {2606.26170},
  archivePrefix = {arXiv},
  primaryClass  = {q-bio.PE}
}

@article{RipleyCaswell2008,
  author  = {Ripley, Bonnie J. and Caswell, Hal},
  title   = {Contributions of growth, stasis, and reproduction to fitness in brooding and broadcast spawning marine bivalves},
  journal = {Population Ecology},
  year    = {2008},
  volume  = {50},
  number  = {2},
  pages   = {207--214},
  doi     = {10.1007/s10144-008-0075-7}
}

@article{Nielsen2016GreenlandShark,
  author  = {Nielsen, Julius and Hedeholm, Rasmus B. and Heinemeier, Jan and Bushnell, Peter G. and Christiansen, J{\o}rgen S. and Olsen, Jesper and Ramsey, Christopher Bronk and Brill, Richard W. and Simon, Malene and Steffensen, Kirstine F. and Steffensen, John Fleng},
  title   = {{Eye lens radiocarbon reveals centuries of longevity in the Greenland shark ({\it Somniosus microcephalus})}},
  journal = {Science},
  year    = {2016},
  volume  = {353},
  number  = {6300},
  pages   = {702--704},
  doi     = {10.1126/science.aaf1703}
}

@article{George1999BowheadWhale,
  author  = {George, J. C. and Bada, J. and Zeh, J. and Scott, L. and Brown, S. E. and O'Hara, T. and Suydam, R.},
  title   = {{Age and growth estimates of Bowhead whales ({\it Balaena mysticetus}) via aspartic acid racemization}},
  journal = {Canadian Journal of Zoology},
  year    = {1999},
  volume  = {77},
  number  = {4},
  pages   = {571--580},
  doi     = {10.1139/z99-015}
}

@article{FrancisHorn1997OrangeRoughy,
  author  = {Francis, R. I. C. C. and Horn, P. L.},
title = {{Transition zone in otoliths of Orange roughy
(\textit{Hoplostethus atlanticus}) and its relationship
to the onset of maturity}},
  journal = {Marine Biology},
  year    = {1997},
  volume  = {129},
  number  = {4},
  pages   = {681--687},
  doi     = {10.1007/s002270050211}
}

@article{RidgwayRichardsonAustad2011,
  author  = {Ridgway, I. D. and Richardson, C. A. and Austad, S. N.},
  title   = {Maximum shell size, growth rate, and maturation age correlate with longevity in bivalve molluscs},
  journal = {The Journals of Gerontology Series A: Biological Sciences and Medical Sciences},
  year    = {2011},
  volume  = {66A},
  number  = {2},
  pages   = {183--190},
  doi     = {10.1093/gerona/glq172}
}

@article{Bourn1977Aldabra,
  author  = {Bourn, David},
  title   = {Reproductive study of giant tortoises on {A}ldabra},
  journal = {Journal of Zoology},
  year    = {1977},
  volume  = {182},
  number  = {1},
  pages   = {27--38},
  doi     = {10.1111/j.1469-7998.1977.tb04138.x}
}

@book{Gerlach2004GiantTortoises,
  author    = {Gerlach, Justin},
  title     = {Giant Tortoises of the Indian Ocean: The Genus {\it Dipsochelys} Inhabiting the Seychelles Islands and the Extinct Giants of Madagascar and the Mascarenes},
  publisher = {Edition Chimaira},
  address   = {Frankfurt am Main},
  year      = {2004},
  pages     = {207}
}

@article{Finkelstein2010LaysanAlbatross,
  author  = {Finkelstein, Myra E. and Doak, Daniel F. and Nakagawa, Melinda and Sievert, Paul R. and Klavitter, John},
  title   = {Assessment of demographic risk factors and management priorities: impacts on juveniles substantially affect population viability of a long-lived seabird},
  journal = {Animal Conservation},
  year    = {2010},
  volume  = {13},
  number  = {2},
  pages   = {148--156},
  doi     = {10.1111/j.1469-1795.2009.00311.x}
}

@article{S,
  author  = {David Singer}, 
  title   = {Stable orbits and bifurcation of maps of the interval},
  journal = {Applied Mathematics},
  year    = 1978,
  number  = 2,
  pages   = {260--267},
  volume  = 35,
}

@book{Stearns1992,
  author    = {Stearns, Stephen C.},
  title     = {The Evolution of Life Histories},
  publisher = {Oxford University Press},
  address   = {Oxford and New York},
  year      = {1992},
  edition   = {1},
  isbn      = {9780198577416},
  isbn13    = {9780198577416},
  isbn10    = {0198577419},
  pages     = {249},
  language  = {English}
}

@article{SteinerTuljapurkar2020,
  author  = {Steiner, Ulrich K. and Tuljapurkar, Shripad},
  title   = {Drivers of Diversity in Individual Life Courses:
             Sensitivity of the Population Entropy of a {M}arkov Chain},
  journal = {Theoretical Population Biology},
  year    = {2020},
  volume  = {133},
  pages   = {159--167},
  doi     = {10.1016/j.tpb.2020.01.003}
}

@article{Tuljapurkar1982,
  author  = {Tuljapurkar, Shripad D.},
  title   = {Why Use Population Entropy? {I}t Determines the Rate of Convergence},
  journal = {Journal of Mathematical Biology},
  year    = {1982},
  volume  = {13},
  number  = {3},
  pages   = {325--337},
  doi     = {10.1007/BF00276067}
}
\end{document}